\documentclass[11pt,a4paper]{article}%
\usepackage{amsfonts}
\usepackage{amsmath}
\usepackage{amssymb}
\usepackage{graphicx}
\usepackage{geometry}%
\usepackage{xcolor}
\providecommand{\U}[1]{\protect\rule{.1in}{.1in}}
\newtheorem{theorem}{Theorem}

\newtheorem{corollary}{Corollary}

\newtheorem{remark}{Remark}

\newenvironment{proof}[1][Proof]{\textbf{#1.} }{\  \rule{0.5em}{0.5em}}
\makeatletter
\def \@removefromreset#1#2{\let \@tempb \@elt
\def \@tempa#1{@&#1}\expandafter \let \csname @*#1*\endcsname \@tempa
\def \@elt##1{\expandafter \ifx \csname @*##1*\endcsname \@tempa \else
\noexpand \@elt{##1}\fi}     \expandafter \edef \csname cl@#2\endcsname{\csname cl@#2\endcsname}     \let \@elt \@tempb
\expandafter \let \csname @*#1*\endcsname \@undefined}

\@removefromreset{equation}{section}

\@removefromreset{theorem}{section}
\makeatother
\begin{document}

\title{Quantifying Bell nonlocality of a pure two-qudit state

via its entanglement}

\author{Elena R. Loubenets$^{1}$, Sergey Kuznetsov$^{2}$ and Louis Hanotel$^{1}$ \\
\\
$^{1}$National Research University Higher School of Economics, \\Moscow 101000, Russia\\$^{2}$Steklov Mathematical Institute of Russian Academy of Sciences, \\Moscow 119991, Russia}

\maketitle

\begin{abstract}
For the maximal violation of  all Bell inequalities by an arbitrary  pure two-qudit state of any dimension, we derive a new lower bound  expressed via the  concurrence of this pure state. This new lower bound and the upper bound on the maximal Bell violation, found in [J. Phys. A: Math. Theor. 55,
285301 (2022)] and also expressed via the concurrence,  analytically quantify Bell nonlocality of a pure two-qudit state via its entanglement, in particular,  prove explicitly that entanglement of a pure two-qudit state is necessary and sufficient for its Bell nonlocality. By re-visiting the pure two-qubit case, we also find and rigorously prove the new results on correlation properties of an arbitrary  pure two-qubit state.
\end{abstract}

\section{Introduction}

Bell
nonlocality \cite{01,02} of a multipartite quantum state -- in the sense of its violation of a Bell
inequality\footnote{For the general framework on Bell inequalities, either on correlation functions or on joint probabilities or of a more general type, see in Ref. 
\cite{04}.} -- has been analyzed by many authors, for details see articles 
\cite{05, 06, 07} and references therein.

 It is well known that every separable quantum state admits a local hidden variable (LHV) model and does not, therefore,  violate any  of 
 Bell inequalities \cite{04}. This implies that every  nonlocal quantum state is entangled. The converse is not, in general, true and, in 1989, Werner \cite{08} presented
mixed entangled two-qudit states, each admitting an LHV model under all projective local measurements with any numbers of settings at each of two  sites. These entangled mixed states do not violate any of bipartite 
Bell inequalities and are,  therefore, fully Bell local in terminology of Ref. \cite{07}. 

However, as proved further in 1991 by Gisin \cite{09} for a
pure two-qubit case and in 1992 by Gisin and Peres \cite{010}\  for 
 a general
two-qudit case, every pure entangled two-qudit state is Bell nonlocal. In both these articles,
the proofs of Bell nonlocality of a pure entangled two-qudit state $\rho
_{\psi}=|\psi\rangle\langle\psi|$ are based on constructing for a considered state of some 
specific qudit observables $A_{i},B_{j},$ $i,j=1,2$, 
for which the Clauser--Horne--Shimony--Holt (CHSH) inequality \cite{011} is
violated:%
\begin{equation}
\left\vert \text{ }\langle A_{1}\otimes B_{1}\rangle_{\rho_{\psi}}+\langle
A_{1}\otimes B_{2}\rangle_{\rho_{\psi}}+\langle A_{2}\otimes B_{1}%
\rangle_{\rho_{\psi}}-\langle A_{2}\otimes B_{2}\rangle_{\rho_{\psi}%
}\right\vert >2. \label{1_}%
\end{equation}


Further, in 1995, for the maximal violation by a general two-qubit state $\rho$ of the CHSH
inequality in case of traceless qubit observables, Horodeckis \cite{012} derived the precise expression 
\begin{align}
 \max_{A_{i},B_{j}} & \left\vert \text{ }\langle A_{1}\otimes B_{1}%
\rangle_{\rho}+\langle A_{1}\otimes B_{2}\rangle_{\rho}+\langle A_{2}\otimes
B_{1}\rangle_{\rho}-\langle A_{2}\otimes B_{2}\rangle_{\rho}\right\vert
\label{2}\\
&  =2M_{chsh}(\rho),\nonumber\\
M_{chsh}(\rho) &  :=\sqrt{u_{1}^{2}(\rho)+u_{2}^{2}(\rho)},\nonumber
\end{align}
where $u_{1}(\rho),u_{2}(\rho)$ are two
largest singular values of the linear
operator $T_{\rho}$ on $\mathbb{R}^{3}$ specified by its matrix
representation in the standard basis in $\mathbb{R}^{3}$:
\begin{equation}
T_{\rho}^{(ij)}:=\mathrm{tr}[\rho(\sigma_{i}\otimes\sigma_{j}%
)]\in\mathbb{R},\text{ \ \ }i,j=1,2,3,\label{3}%
\end{equation}
and referred to as the correlation matrix of a two-qubit state $\rho$. Here, $\sigma_{k}$, $k=1,2,3,$ are
the Pauli operators on $\mathbb{C}^{2}.$ 
Relation (\ref{2}) implies that, in case of traceless qubit observables
at each of two sites, a general two-qubit state $\rho$ violates the CHSH inequality
iff parameter $M_{chsh}(\rho )>1$. 

The results for a two-qubit state in \cite{010,012} and the result for a two-qudit state in \cite{011} do not, however,  establish any quantitative relation between entanglement and nonlocality of a pure two-qudit state which could explicitly imply that entanglement of this state is sufficient for its nonlocality. Moreover, to our knowledge, up to now, such a quantitative relation has not been reported 
 in the literature. 

The aim of the present article is just to find such a quantitative relation between entanglement and nonlocality of a pure two-qudit state of an arbitrary dimension. 


The article is organized as follows. In Section 2, we introduce the measures for entanglement and nonlocality used in the present paper. In Section 3, by re-visiting a two-qubit case, we find and rigorously prove the new results on correlation properties and Bell nonlocality of an arbitrary  pure two-qubit state. In Section 4, for an arbitrary  pure two-qudit, we find a new lower bound on its maximal  violation of all Bell inequalities which is expressed via the concurrence of this state. This new lower bound  and  the upper bound on the maximal  Bell violation, derived recently in \cite{019},  analytically quantify Bell nonlocality of a pure two-qudit state via its entanglement, in particular,  prove explicitly  that every entangled two-qudit state is Bell nonlocal. In Section 5, we summarize the main results of this article. 

\section{Preliminaries} 

In this Section we introduce the measures for entanglement and nonlocality used in the present paper.

Recall that entanglement of a bipartite state $\rho$ can be quantified
via bipartite entanglement measures, in particular, via negativity and concurrence. For a pure two-qudit state $\rho_{\psi
}=|\psi\rangle\langle\psi|$ on $\mathbb{C}^{d_{1}}\mathbb{\otimes C}^{d_{2}}$, its 
negativity $\mathcal{N}_{\rho_{\psi}}$ and the unnormalized concurrence
$\mathrm{C}_{\rho_{\psi}}$ satisfy the relation (see relation (38) in \cite{019})%
\begin{equation}
0\leq2\sqrt{2}\text{ }\mathcal{N}_{\rho_{\psi}}\leq r_{sch}^{(\psi)}%
(r_{sch}^{(\psi)}-1)\mathrm{C}_{\rho_{\psi}}, \label{4__}%
\end{equation}
where $1\leq r_{sch}^{(\psi)}\leq d=\min\{d_{1},d_{2}\}$ is the Schmidt rank
of a pure state $\rho_{\psi
}=|\psi\rangle\langle\psi|$ and the unnormalized concurrence 
\cite{014,014_}%
\begin{align}
\mathrm{C}_{\rho_{\psi}}  &  :=\sqrt{2\left(  1-\text{\textrm{tr}}[\tau
_{j}^{2}(\rho_{\psi})]\right)  },\label{5_}\\
0  &  \leq\mathrm{C}_{\rho_{\psi}}\leq\sqrt{\frac{2(d-1)}{d}},\nonumber
\end{align}
can be \emph{easily calculated}. A pure two-qudit state is entangled iff $\mathrm{C}_{\rho_{\psi}}>0$. Here, $\tau_{j}(\rho_{\psi}),\ j=1,2,$ are
states on $\mathbb{C}^{d_{1}}$ and $\mathbb{C}^{d_{2}},$ respectively, reduced from state
$\rho_{\psi}$ on $\mathbb{C}^{d_{1}}\otimes\mathbb{ C}^{d_{2}},$  and by the Schmidt theorem \textrm{tr}$[\tau_{1}^{2}(\rho_{\psi
})]=$ \textrm{tr}$[\tau_{2}^{2}(\rho_{\psi})]$. 

For a general two-qudit state $\rho$  concurrence $\mathrm{C}_{\rho}$ is
given by the infimum \cite{015,016}
\begin{equation}
\mathrm{C}_{\rho}:=\inf_{\{\alpha_{i},\psi_{i}\}} \sum\alpha_{i}\mathrm{C}_{|\psi_{i}
\rangle\langle\psi_{i}|}, \label{6_}%
\end{equation}
over all possible convex decompositions of a state $\rho=\sum\alpha_{i}%
|\psi_{i}\rangle\langle\psi_{i}|,$ $\sum\alpha_{i}=1,$ $\alpha_{i}>0,$ via
pure states.

The situation with quantifying Bell nonlocality is quite different. To our knowledge, among all possible nonlocality measures suggested\footnote{On nonlocality measures different to (\ref{7_}), see articles \cite{nlm5,nlm2,nlm3, nlm4, nlm1} and references
therein.} in the literature, there is only one which is \emph{consistent with
the above definition of the Bell nonlocality for any }$N$\emph{-partite state
}$\rho,$\emph{ pure or mixed}. This nonlocality measure, introduced and
discussed in \cite{05, 07,019, 018}, is given by 
\begin{equation}
1\leq \Upsilon_{\rho}:=\sup_{\mathfrak{B}\neq 0}\Upsilon_{\mathfrak{B}}^{(\rho
)},\label{7_}%
\end{equation}
where supremum is taken  over all possible Bell inequalities\footnote {Either on correlation functions or on joint probabilities   or of a more general type and   with any number of measurement settings at each of sites.} and  parameter 
$\Upsilon_{\mathfrak{B}}^{(\rho)}$ constitutes the maximal violation   
of the Bell inequality,  specified by a Bell functional  $\mathfrak{B(\cdot)}$, by an $N$-partite state $\rho$ and  is defined via maximizing  over all possible settings and outcomes the ratio of the absolute value of this functional on a state $\rho$ to its maximal absolute value within the LHV model.
For example, for the  Clauser--Horne inequality \cite{CH} the latter LHV value is equal to 1. For details, see the general formalism developed in \cite{05} (for its  short version, see  Section 2 in  \cite{019}). 

\emph{Definition (\ref{7_}) implies  that an $N$-partite quantum state $\rho$ is
fully Bell local, that is, does not violate any Bell inequality, iff $\Upsilon_{\rho}=1$ and nonlocal iff $\Upsilon_{\rho}>1.$}

Clearly, any separable quantum state $\rho_{sep}$ is fully Bell local and,
by the analytical upper bound  (53) on $\Upsilon_{\rho}$ derived in
\cite{05} and expressed via the specific extensions of a state $\rho$, there exist
entangled $N$-partite quantum states $\rho$ that are fully Bell local.
Based on this analytical upper bound on $\Upsilon_{\rho}\geq1$ valid for any
$N$-partite state $\rho$ , there were further \cite{05,019,018} specified several upper bounds on the maximal Bell violation $\Upsilon_{\rho}\geq1$ in terms of the Hilbert space characteristics of this state. In particular, the
upper bound \cite{019} in terms of the concurrence: %
\begin{equation}
\Upsilon_{\rho}\leq1+\sqrt{2d(d-1)}\text{ }\mathrm{C}_{\rho},\text{ \ \ }d=\min
\{d_{1},d_{2}\},\label{8__}%
\end{equation}
which is valid for every two-qudit state $\rho$  on $\mathbb{C}^{d_{1}}\otimes\mathbb{
C}^{d_{2}}$, pure or mixed, and also, the upper bound \cite{019}
\begin{equation}
\Upsilon_{\rho_{\psi}}\leq1+\sqrt{2r_{sch}^{(\psi)}(r_{sch}^{(\psi)}-1)}\text{
}\mathrm{C}_{\rho_{\psi}},\label{9__}%
\end{equation}
true for every pure two-qudit  state $\rho_{\psi}$ with the Schmidt rank   $1\leq r_{sch}^{(\psi)}\leq
d=\min\{d_{1},d_{2}\}$. 
Since $ \mathrm{C}_{\rho}\leq\sqrt{\frac{2(d-1)}{d}},$ the upper bound (\ref{8__}) is
tighter than the upper bound%
\begin{equation}
\Upsilon_{\rho}\leq2\min\{d_{1},d_{2}\}-1\label{10_},
\end{equation}
derived earlier in \cite{05}. 

The upper bound (\ref{8__}) explicitly implies that
entanglement of a general two-qudit state is necessary for its Bell nonlocality:
$\Upsilon_{\rho}>1$ $\Rightarrow$ $\mathrm{C}_{\rho}>0$, and is important for the entanglement certification and quantification \cite{Moroder2013,Liang2011,app}.

\section{Two-qubit case re-visited}

Different aspects of correlation properties of a pure two-qubit state have been discussed by many authors, see, for example, articles \cite{020,Batle,Gamel} and references therein. 

In particular, in \cite{020}, the authors argue to prove that, for any pure state $\rho_{\psi}=|\psi\rangle\langle\psi|$ on $\mathbb{C}^{2}\otimes\mathbb{
C}^{2},$ the state parameter $M_{chsh}(\rho_{\psi})$ in (\ref{2}) has the form
\begin{equation}
M_{chsh}(\rho_{\psi})=\sqrt{1+\mathrm{C}_{\rho_{\psi}}^{2}}\ ,\label{Mchsh}%
\end{equation}
where $\mathrm{C}_{\rho_{\psi}}$ is concurrence (\ref{5_}) of a pure two-qubit state $\rho_{\psi}$.   Though this equality was further referred to and used in many other papers, the
proof of this equality in \cite{020} is given for a pure two-qubit state which has a diagonalizable correlation matrix (\ref{3}) and admits the Schmidt decomposition of some special form. 

This is definitely not the case for every pure two-qubit state, for details see below our Remark 1 and our discussion after Corollary 1.

In this Section, we derive the new results on correlation properties of a  pure two-qubit state, implying, in particular, the rigorous proof of relation (\ref{Mchsh}) for an arbitrary two-qubit state.  
To our knowledge, these our new results on correlation properties of a pure two-qubit state  have not been reported in the literature.

For a pure two-qubit state $\rho_{\psi}=|\psi\rangle\langle\psi|,$
 consider its Pauli representation
\begin{align}
&\rho_{\psi}    =\frac{1}{4}\left[  \mathbb{I}\otimes\mathbb{I}+\left( \vec{r}_{1}(\rho_{\psi})\cdot\vec{\sigma}\right)
\otimes\mathbb{I}+\mathbb{I}\otimes\left(\vec{r}_{2}%
(\rho_{\psi})\cdot\vec{\sigma}\right)  +\sum\limits_{i,j=1}%
^{3}T_{\rho_{\psi}}^{(ij)}(\sigma_{i}\otimes\sigma_{j})\right]  ,\label{5}\\
&\vec{r}\cdot\vec{\sigma}    :=\sum\limits_{i=1}%
^{3}r^{(i)}\sigma_{i},\     \vec{\sigma}:=(\sigma
_{1},\sigma_{2},\sigma_{3}),\nonumber
\end{align}
where $\sigma_{k},$ $k=1,2,3,$ are the Pauli operators, $T_{\rho_{\psi}}$ is
the correlation matrix (\ref{3}) and
\begin{align}
\vec{r}_{1}(\rho_{\psi})  &  :=\mathrm{tr}\,[\rho_{\psi
}(\vec{\sigma}\otimes\mathbb{I})]=\mathrm{tr}\,[\tau_{1}%
(\rho_{\psi})\vec{\sigma}\mathbb{]},\text{ \ \ \ }%
\vec{r}_{2}(\rho_{\psi}):=\mathrm{tr}\,\mathbb{[}\rho_{\psi
}(\mathbb{I\otimes}\vec{\sigma}\mathbb{)]=}\text{ }\mathrm{tr}\,
[\tau_{2}(\rho_{\psi})\vec{\sigma}\mathbb{]}\text{,}\label{6}\\
\vec{r}_{j}(\rho_{\psi})  &  \in\mathbb{R}^{3},\text{
\ \ }j=1,2,\nonumber
\end{align}
are the Bloch vectors of the qubit states%
\begin{equation}
\label{7}
\tau_{j}(\rho_{\psi})=\frac{\mathbb{I}+\vec{r}_{j}(\rho_{\psi
})\cdot\vec{\sigma}}{2},\text{ \ \ \ }\mathrm{tr}[\tau_{j}^{2}%
(\rho_{\psi})]=\frac{1}{2}\left(  1+\left\Vert \vec{r}_{j}%
(\rho_{\psi})\right\Vert _{\mathbb{R}^{3}}^{2}\right)  ,\text{ \ \ \ }j=1,2,
\end{equation}
reduced from a state $\rho_{\psi}$ on $\mathbb{C}^{2}\otimes\mathbb{C}^{2}.$
By the Schmidt theorem $\left\Vert \vec{r}_{1}(\rho_{\psi
})\right\Vert _{\mathbb{R}^{3}}=\left\Vert \vec{r}_{2}(\rho_{\psi
})\right\Vert _{\mathbb{R}^{3}}.$

\begin{remark}

In Ref. \cite{020}, the authors claim that, for any pure two-qubit state, the correlation matrix (\ref{3}) is 
diagonalizable. However, this claim is mistaken. As an example, let us take the separable pure state $\rho_{\psi_{sep}}$ with $|\psi_{sep}\rangle
=\frac{1}{\sqrt{2}}|0\rangle\otimes\left(  \text{ }|0\rangle+i|1\rangle
\right)  $. For this pure two-qubit state,  the Bloch vectors in (\ref{5}) are equal to  $\vec{r}_{1}=(0,0,1)$ and $\vec{r}%
_{2}=(0,1,0)$ and 
the correlation
matrix, having only one nonzero element $T_{\rho_{\psi_{sep}}%
}^{(32)}=1$, has
eigenvalue $0$ with the algebraic multiplicity $3$ and
the geometric multiplicity $2$. The latter implies that, for this pure two-qubit
state, the correlation matrix $T_{\rho_{\psi_{sep}}}$ is non-diagonalizable. 
\end{remark}

The Pauli representation (\ref{5}) and the purity relation $\rho_{\psi}%
^{2}=\rho_{\psi}$ imply
\begin{align}
3  &  =\left\Vert \vec{r}_{1}\right\Vert _{\mathbb{R}^{3}}%
^{2}\ +\ \left\Vert \vec{r}_{2}\right\Vert _{\mathbb{R}^{3}}%
^{2}\ +\ \mathrm{tr[}T_{\rho_{\psi}}^{\dagger}T_{\rho_{\psi}}],\label{8}\\
\vec{r}_{1}  &  =T_{\rho_{\psi}}\vec{r}_{2},\text{
\ \ \ }\vec{r}_{2}=T_{\rho_{\psi}}^{\dagger}\vec{r}
_{1},\label{9}\\
T_{\rho_{\psi}}^{(ij)}  &  ={r}_{1}^{(i)}r_{2}^{(j)}-\frac
{1}{2}\sum\limits_{k,m,n,p}T_{\rho_{\psi}}^{(km)}T_{\rho_{\psi}}%
^{(np)}\varepsilon_{kni}\varepsilon_{mpj}. \label{10}%
\end{align}
By relation (\ref{3}), for a separable pure two-qubit state, the correlation matrix has elements $T_{\rho_{\psi}}^{(ij)}   = {r}_{1}^{(i)}r_{2}^{(j)}$, for which the second term in (\ref{10}) is equal to zero for any $i,j=1,2,3$.

By (\ref{7}), the concurrence (\ref{5_}) of a pure two-qubit state $\rho_{\psi}$
takes the form\footnote{This form is also true for the normalized concurrence
of any pure bipartite quantum state on $\mathbb{C}^{d_{1}}\otimes
\mathbb{C}^{d_{2}},$ see Theorem 5 in \cite{21}.}%
\begin{equation}
\mathrm{C}_{\rho_{\psi}}=\sqrt{1-\gamma_{\rho_{\psi}}^{2}},\text{
\ \ \ }\gamma_{\rho_{\psi}}=\left\Vert \vec{r}_{1}(\rho_{\psi
})\right\Vert _{\mathbb{R}^{3}}=\left\Vert \vec{r}_{2}(\rho_{\psi
})\right\Vert _{\mathbb{R}^{3}}. \label{11}%
\end{equation}
This and relation (\ref{8}) imply%
\begin{equation}
\mathrm{tr[}T_{\rho_{\psi}}^{\dagger}T_{\rho_{\psi}}]=1+2\mathrm{C}%
_{\rho_{\psi}}^{2}. \label{12}%
\end{equation}
Taking into account relations (\ref{8})--(\ref{12}),  we find the following new result.

\begin{theorem}
For any pure two-qubit state $\rho_{\psi}=|\psi\rangle\langle\psi|,$
$|\psi\rangle\in\mathbb{C}^{2}\otimes\mathbb{C}^{2},$ the determinant of the correlation matrix
$T_{\rho_{\psi}}$ is non-positive  and is expressed via its concurrence as
\begin{equation}
\det T_{\rho_{\psi}}=-\mathrm{C}_{\rho_{\psi}}^{2}. \label{13}%
\end{equation}

\end{theorem}

\begin{proof}
Taking in (\ref{10}) into account that
\begin{equation}
\sum\limits_{k,m,n,p}T_{\rho_{\psi}}^{(km)}T_{\rho_{\psi}}^{(np)}\varepsilon_{kni}%
\varepsilon_{mpj}=2A_{T_{\rho_{\psi}}}^{(ij)}, \label{14}%
\end{equation}
where $A_{T_{\rho_{\psi}}}^{(ij)}$ is the cofactor of the matrix element
$T_{\rho_{\psi}}^{(ij)}$, we have%
\begin{equation}
T_{\rho_{\psi}}^{(ij)}=r_{1}^{(i)}r_{2}^{(j)}-A_{T_{\rho_{\psi}}}^{(ij)}.
\label{15}%
\end{equation}
Multiplying the left-hand and the right-hand sides of (\ref{15}) by the matrix
element $T_{\rho_{\psi}}^{(ij)}$, summing further over $i,j$ and taking into
account (\ref{9}), we derive%
\begin{equation}
\sum\limits_{i,j=1,2,3}\left( T_{\rho_{\psi}}^{(ij)}\right)  ^{2}=\left\Vert
\vec{r}_{1}\right\Vert ^{2}_{\mathbb{R}^{3}} -\sum\limits_{i,j=1,2,3}A_{T_{\rho
_{\psi}}}^{(ij)}T_{\rho_{\psi}}^{(ij)}, \label{16}%
\end{equation}
where the terms in the left-hand and the right-hand sides constitute
\begin{align}
\sum\limits_{i,j=1,2,3}\left(  T_{\rho_{\psi}}^{(ij)}\right)  ^{2}  &
=\mathrm{tr[}T_{\rho_{\psi}}^{\dagger}T_{\rho_{\psi}}],\label{17}\\
\sum\limits_{j=1,2,3}A_{T_{\rho_{\psi}}}^{(ij)}T_{\rho_{\psi}}^{(ij)}  &
=\det T_{\rho_{\psi}}.\nonumber
\end{align}
This and equality (\ref{16}) imply
\begin{equation}
\mathrm{tr[}T_{\rho_{\psi}}^{\dagger}T_{\rho_{\psi}}]=\left\Vert
\vec{r}_{1}\right\Vert ^{2}_{\mathbb{R}^{3}}-3\det T_{\rho_{\psi}}. \label{18}%
\end{equation}
By (\ref{11}), (\ref{12}) and (\ref{18}):%
\begin{equation}
1+2\mathrm{C}_{\rho_{\psi}}^{2}=1-\mathrm{C}_{\rho_{\psi}}^{2}-3\det
T_{\rho_{\psi}}\text{ \ \ }\Leftrightarrow\text{ \ }\mathrm{C}_{\rho_{\psi}%
}^{2}=-\det T_{\rho_{\psi}}.
\label{19}%
\end{equation} 
This proves the statement.
\end{proof}
\\

The new result of the  following statement explicitly expresses the singular values of the
correlation matrix $T_{\rho_{\psi}}$ of an arbitrary pure two-qubit state $\rho_{\psi}$ and,
correspondingly, the eigenvalues of the positive matrices $T_{\rho_{\psi}
}^{\dagger}T_{\rho_{\psi}}$ and $T_{\rho_{\psi}} T_{\rho_{\psi}}
^{\dagger} $ via the concurrence of this pure state. 

\begin{theorem}
For any pure two-qubit state $\rho_{\psi}=|\psi\rangle\langle\psi|,$
$|\psi\rangle\in\mathbb{C}^{2}\otimes\mathbb{C}^{2},$ the singular values $u_{j}(\rho_{\psi}),$
$j=1,2,3,$ of its correlation matrix $T_{\rho_{\psi}}$ are equal to
\begin{equation}
u_{1}(\rho_{\psi})=1,\text{ \ \ \ }u_{2}(\rho_{\psi})=u_{3}(\rho_{\psi
})=\mathrm{C}_{\rho_{\psi}}, \label{20}%
\end{equation}
where $\mathrm{C}_{\rho_{\psi}}$ is the concurrence of $\rho_{\psi}$
and, in (\ref{5}), the Bloch vectors $\vec{r}_{2}$ and $\vec{r}_{1}$ constitute the eigenvectors of the positive operators $T_{\rho_{\psi}}^{\dagger}T_{\rho_{\psi}}$ and $T_{\rho_{\psi}}T_{\rho_{\psi}}^{\dagger}$,  respectively, corresponding to their eigenvalue $1$. 

\end{theorem}

\begin{proof}
Relations (\ref{9}) imply
\begin{equation}
\vec{r}_{2}=T_{\rho_{\psi}}^{\dagger}T_{\rho_{\psi}}%
\vec{r}_{2},\text{ \ \ }\vec{r}_{1}=T_{\rho_{\psi}%
}T_{\rho_{\psi}}^{\dagger}\vec{r}_{1}. \label{21}%
\end{equation}
Therefore, the positive operators $T_{\rho_{\psi}}^{\dagger}T_{\rho_{\psi}}$
and $T_{\rho_{\psi}}T_{\rho_{\psi}}^{\dagger}$ have the eigenvalue  $1$  and the Bloch vectors 
$\vec{r}_{2}$ and $\vec{r}_{1}$ 
are their eigenvectors   corresponding to this eigenvalue. Thus, one of the
singular values of $T_{\rho_{\psi}}$ is equal to $1. $ Let it be $u_{1}%
(\rho_{\psi})=1.$ Taking this into account in the equality $\left\vert \det
T_{\rho_{\psi}}\right\vert =u_{1}(\rho_{\psi})u_{2}(\rho_{\psi})u_{3}%
(\rho_{\psi})$, we have
\begin{equation}
\left\vert\det T_{\rho_{\psi}}\right\vert =u_{2}(\rho_{\psi}%
)u_{3}(\rho_{\psi}). \label{22}%
\end{equation}
This and relation (\ref{13}) imply
\begin{equation}
\mathrm{C}_{\rho_{\psi}}^{2}=u_{2}(\rho_{\psi})u_{3}(\rho_{\psi}). \label{23}%
\end{equation}
From relations (\ref{12}), (\ref{23}) and $u_{1}=1$ it follows%
\begin{align}
1+u_{2}^{2}+u_{3}^{2}  &  =1+2u_{2}u_{3}\label{24}\\
&  \Leftrightarrow\nonumber\\
\left(  u_{2}-u_{3}\right)  ^{2}  &  =0\text{ \ \ \ }\Leftrightarrow\text{
\ \ }u_{2}=u_{3}.\nonumber
\end{align}
By (\ref{23}), this implies 
\begin{equation}
u_{2}(\rho_{\psi})=u_{3}(\rho_{\psi})=\mathrm{C}_{\rho_{\psi}}\leq1. \label{25}
\end{equation} This completes the proof of the statement.
\end{proof}
\\

Relations (\ref{20})  and (\ref{2}) imply.

\begin{corollary}

 For any pure two-qubit state $\rho_{\psi} $,  the state parameter $M_{chsh}(\rho_{\psi})$, defining in (\ref{2}) the maximal violation of the CHSH inequality for traceless qubit observables, is equal to 
 
\begin{equation}
M_{chsh}(\rho_{\psi})=\sqrt{1+\mathrm{C}_{\rho_{\psi}}^{2}}. \label{30}%
 \end{equation}
 
 \end{corollary}
 
As mentioned at the beginning of this Section, relation (\ref{30}) was first presented 
in \cite{020} but proved there   only for 
a pure two-qubit state $\rho_{\psi}$ with vector  $|\psi\rangle\in\mathbb{C}^{2}\otimes\mathbb{ C}^{2}$
admitting the Schmidt decomposition of the very special form
$ |\psi\rangle=\sqrt{\lambda_{0}}|00\rangle+\sqrt{\lambda_{1}}|11\rangle\ $  where the correlation matrix (\ref{3}) is diagonal.
However, both conditions do not need to hold for every two-qubit state. 
Take, for example, the
pure two-qubit state, which we   consider in Remark 1 and show that the correlation matrix of this state is not diagonalizable. This state has
projections on both sub-spaces 
 of  $\mathbb{C}^{2}\otimes \mathbb{ C}^{2}$ --
symmetric and antisymmetric, and, therefore, cannot  admit the
Schmidt decomposition of the above special form.

Taking into account relations  (\ref{9__}), (\ref{30}) and that, for a state $\rho_{\psi}$, the maximal violation (\ref{7_})   of all Bell inequalities is not less  than its maximal  violation $M_{chsh}(\rho_{\psi})$ of the CHSH inequality, we come to the following statement.

\begin{theorem}
For every pure state $\rho_{\psi}=|\psi\rangle\langle\psi|$ on $\mathbb{C}%
^{2}\otimes\mathbb{C}^{2}$,  the maximal violation (\ref{7_}) of all Bell inequalities, either on correlation functions or on joint probabilities and with any number of settings and outcomes at each of two sites, admits the
bounds
\begin{equation}
\sqrt{1+\mathrm{C}_{\rho_{\psi}}^{2}}\leq \Upsilon_{\rho_{\psi}}\leq
1+\sqrt{2r_{sch}^{(\psi)}(r_{sch}^{(\psi)}-1)}\mathrm{C}_{\rho_{\psi}}\leq
1+2\mathrm{C}_{\rho_{\psi}},\label{31}%
\end{equation}
where $r_{sch}^{(\psi)}=1,2$ is the Schmidt rank of $|\psi\rangle$ and
$\mathrm{C}_{\rho_{\psi}}=0\Leftrightarrow r_{sch}^{(\psi)}=1.$
\end{theorem}

Since any quantum state $\rho$ is nonlocal iff $\Upsilon_{\rho}>1$  and entangled iff $\mathrm{C}_{\rho}>0$ (see in Section 2), relation (\ref{31}) explicitly proves  that entanglement of a pure two-qubit state is necessary and sufficient for its Bell nonlocality. 

 \section{General two-qudit case}

In order to find a lower bound on the maximal violation (\ref{7_}) for a pure
state $\rho_{\psi}=|\psi\rangle\langle\psi|$ on $\mathbb{C}^{d_{1}%
}\otimes\mathbb{ C}^{d_{2}},$ we recall the result in Ref. \cite{010}.

Let, for a pure state $|\psi\rangle\in\mathbb{C}^{d_{1}}\otimes\mathbb{C}^{d_{2}}$,
the Schmidt decomposition read
\begin{equation}
|\psi\rangle=\sum_{1\leq k\leq r_{sch}^{(\psi)}}\sqrt{\lambda_{k}(\psi)\text{
}}|e_{k}^{(1)}\rangle\otimes|e_{k}^{(2)}\rangle,\text{ \ \ }\sum_{1\leq k\leq
r_{sch}^{(\psi)}}\lambda_{k}(\psi)=1. \label{32}\ 
\end{equation}
Here: (i) $\lambda_{k_{1}}\geq\lambda_{k_{2}}\geq
...\geq\lambda_{r_{sch}^{(\psi)}}>0$  are the nonzero eigenvalues, taken according
to their multiplicity, of states on $\mathbb{C}^{d_{j}},$ $j=1,2$, reduced from a pure state $|\psi\rangle\langle\psi|$ on
$\mathbb{C}^{d_{1}}\otimes\mathbb{C}^{d_{2}}$; (ii)  $|e_{k}^{(j)}\rangle\in\mathbb{C}^{d_{j}},$ $j=1,2,$ 
are the unit eigenvectors of these reduced states and (iii) $1\leq r_{sch}^{(\psi)}\leq
d=\min\{d_{1},d_{2}\}$ is the Schmidt rank of vector $|\psi\rangle
\in\mathbb{C}^{d_{1}}\otimes\mathbb{C}^{d_{2}}.$

Then, for observables $\widetilde{A}_{i},$ $\widetilde{B}_{j},$ $i,j=1,2,$
chosen in \cite{010} specifically for a state (\ref{32}), the CHSH
combination of the product expectations is equal to \cite{010}:%
\begin{align}
&  \left\vert \text{ }\langle\widetilde{A}_{1}\otimes\widetilde{B}_{1}%
\rangle_{\rho_{\psi}}+\langle\widetilde{A}_{1}\otimes\widetilde{B}_{2}%
\rangle_{\rho_{\psi}}+\langle\widetilde{A}_{2}\otimes\widetilde{B}_{1}%
\rangle_{\rho_{\psi}}-\langle\widetilde{A}_{2}\otimes\widetilde{B}_{2}%
\rangle_{\rho_{\psi}}\right\vert \label{33}\\
&  =2\left\{  \beta(\psi)+\sqrt{\left(  \text{ }1-\beta(\psi)\right)
^{2}+K^{2}(\psi)}\right\}  \nonumber
\end{align}
where $\beta(\psi)=\lambda_{r_{sch}^{(\psi)}}$  if the Schmidt rank
$r_{sch}^{(\psi)}\ $is odd  and $\beta(\psi)=0$ if $r_{sch}^{(\psi)}$ is even.
Parameter $K(\psi)=0$ in a separable case $(r_{sch}^{(\psi)}=1)$ and %

\begin{equation}
K(\psi)=2\left(  \sqrt{\lambda_{1}\lambda_{2}}+\sqrt{\lambda_{3}\lambda_{4}%
}+\sqrt{\lambda_{5}\lambda_{6}}+\cdots\right)  \label{35}%
\end{equation}
in a nonseparable case $(r_{sch}^{(\psi)}\geq2).$ The last term in
(\ref{35}) is equal to $2\sqrt{\lambda_{r_{sch}^{(\psi)}-1}\lambda
_{r_{sch}^{(\psi)}}}$ if $r_{sch}^{(\psi)}$ is even and to $2\sqrt
{\lambda_{r_{sch}^{(\psi)}-2}\lambda_{r_{sch}^{(\psi)}-1}}$ if $r_{sch}%
^{(\psi)}$ is odd. 

By results (\ref{33}), (\ref{35}), found in \cite{010}, and the
relation
\begin{equation}
\gamma+\sqrt{\left(  \text{ }1-\gamma\right)  ^{2}+K^{2}}\geq
\sqrt{1+K^{2}},\ \label{37}%
\end{equation}
valid for any non-negative $0\leq\gamma\leq1,$ we derive that, for any pure two-qudit
state $\rho_{\psi}=|\psi\rangle\langle\psi|$ on $\mathbb{C}^{d_{1}%
}\otimes\mathbb{ C}^{d_{2}},$ its maximal violation of the CHSH inequality 
\begin{equation}
\Upsilon_{chsh}^{(\rho_ {\psi})} \geq\sqrt{1+K^{2}(\psi)}\ . \label{38}%
\end{equation}
This inequality and (\ref{7_}) 
 imply the following new result.

\begin{theorem}
Let $\rho_{\psi}=|\psi\rangle\langle\psi|,$ $|\psi\rangle\in$ $\mathbb{C}%
^{d_{1}}\otimes\mathbb{ C}^{d_{2}}$ be an arbitrary pure state with the
Schmidt rank $1\leq r_{sch}^{(\psi)}\leq d=\min\{d_{1},d_{2}\}$ and
$\Upsilon_{\rho_{\psi}}$ be the maximal violation by state $\rho_{\psi}$ of all bipartite Bell inequalities \ref{7_}, either on correlation functions or on joint
probabilities and with any number of settings and outcomes at each of two
sites. Then $\Upsilon_{\rho_{\psi}}=1$ for $r_{sch}^{(\psi)}=1$ and 
\begin{equation}
\Upsilon_{\rho_{\psi}}\geq\sqrt{1+\frac{1}%
{(2r_{sch}^{(\psi)}-3)^{2}}\mathrm{C}_{\rho_{\psi}}^{2}}\label{39_}%
\end{equation}
for all  $r_{sch}^{(\psi)}\geq2$. Here,  $\mathrm{C}_{\rho_{\psi}}$ is concurrence (\ref{5_}) of state $\rho_{\psi}$.

\end{theorem}

\begin{proof}
Clearly, if $r_{sch}^{(\psi)}=1$,  then $\Upsilon_{\rho_{\psi}}=1$, also, $K(\psi
)=\mathrm{C}_{\rho_{\psi}}=0$. In order to prove the lower
bound in (\ref{39_}) for every pure state $\rho_{\psi}=|\psi\rangle\langle\psi|$ on $\mathbb{C}^{d_{1}}\otimes\mathbb{C}^{d_{2}}$, let us evaluate via the concurrence the term $K(\psi)$,
standing in (\ref{38}). If  $r_{sch}^{(\psi
)}\geq2,$ then by relations  (\ref{5_}) and (\ref{32}) we have
\begin{equation}
\mathrm{C}_{\rho_{\psi}}=\sqrt{2\left(  1-\text{\textrm{tr}}[\tau_{j}^{2}(\rho_{\psi
})]\right)  }=2\sqrt{%
{\textstyle\sum\limits_{i<j}}
\lambda_{i}\lambda_{j}}\leq2%
{\textstyle\sum\limits_{i<j}}
\sqrt{\lambda_{i}\lambda_{j}},\label{41}%
\end{equation}
where the equality holds for $r_{sch}%
^{(\psi)}=2$ and, as settled above,  $\lambda_{k_{1}}\geq\lambda_{k_{2}}%
\geq...\geq\lambda_{r_{sch}^{(\psi)}}>0.$ For evaluating the terms in
expression (\ref{35}) defining parameter $K(\psi)$, we note that, for $r_{sch}^{(\psi
)}\geq2,$ there are $(r_{sch}^{(\psi)}-1)$ inequalities of the form
$\sqrt{\lambda_{1}\lambda_{2}}\geq\sqrt{\lambda_{1}\lambda_{k}}$,
$k=2,\dots,r_{sch}^{(\psi)}$, and $(r_{sch}^{(\psi)}-2)$ inequalities of the form $\sqrt{\lambda_{1}\lambda_{2}}\geq\sqrt{\lambda_{2}\lambda_{k}}$,
$k=3,\dots,r_{sch}^{(\psi)}$. Therefore, for any $r_{sch}^{(\psi)}\geq2,$%
\begin{equation}
(2r_{sch}^{(\psi)}-3)\sqrt{\lambda_{1}\lambda_{2}}\geq\sum_{k=2}%
^{r_{sch}^{(\psi)}}\sqrt{\lambda_{1}\lambda_{k}}+\theta(r_{sch}^{(\psi
)}-3)\sum_{k=3}^{r_{sch}^{(\psi)}}\sqrt{\lambda_{2}\lambda_{k}},
\label{42}%
\end{equation}
where $\theta(x)=1$,  for $x\geq0$, and $\theta(x)=0$, otherwise.
Similarly,
for any  $r_{sch}^{(\psi)}\geq4,$
\begin{equation}
(2r_{sch}^{(\psi)}-7)\sqrt{\lambda_{3}\lambda_{4}}\geq\sum_{k=4}%
^{r_{sch}^{(\psi)}}\sqrt{\lambda_{3}\lambda_{k}}+\theta(r_{sch}^{(\psi
)}-5)\sum_{k=5}^{r_{sch}^{(\psi)}}\sqrt{\lambda_{4}\lambda_{k}}\label{43}%
\end{equation}
and, for  any $r_{sch}^{(\psi)}\geq6,$%
\begin{equation}
(2r_{sch}^{(\psi)}-11)\sqrt{\lambda_{5}\lambda_{6}}\geq\sum_{k=6}%
^{r_{sch}^{(\psi)}}\sqrt{\lambda_{5}\lambda_{k}}+\theta(r_{sch}^{(\psi
)}-7)\sum_{k=7}^{r_{sch}^{(\psi)}}\sqrt{\lambda_{6}\lambda_{k}}\label{eq},%
\end{equation}
with the corresponding analogs for other terms for $K(\psi)$ in (\ref{35}).

These inequalities and (\ref{35})
imply that, for  any $r_{sch}^{(\psi)}\geq2,$
\begin{align}
K(\psi)  & \geq\frac{2}{2r_{sch}^{(\psi)}-3}\left(
\sum_{k=2}^{r_{sch}^{(\psi)}}\sqrt{\lambda_{1}\lambda_{k}}+\theta
(r_{sch}^{(\psi)}-3)\sum_{k=3}^{r_{sch}^{(\psi)}}\sqrt{\lambda_{2}\lambda_{k}%
}\right)  \label{45}\\
& +\frac{2\theta(r_{sch}^{(\psi)}-4)}{2r_{sch}^{(\psi)}-7}\left(  \sum
_{k=4}^{r_{sch}^{(\psi)}}\sqrt{\lambda_{3}\lambda_{k}}+\theta(r_{sch}^{(\psi
)}-5)\sum_{k=5}^{r_{sch}^{(\psi)}}\sqrt{\lambda_{4}\lambda_{k}}\right)
+\cdots.\nonumber
\end{align}
Therefore,  for $r_{sch}^{(\psi)}%
\geq2,$%
\begin{equation}
K(\psi)\geq\frac{2}{2r_{sch}^{(\psi)}-3}\sum_{1\leq i<j\leq r_{sch}^{(\psi)}%
}\sqrt{\lambda_{i}\lambda_{j}}.\label{46}%
\end{equation}
In view of (\ref{41}), this implies, that, for all $r_{sch}^{(\psi)}\geq2,$
\begin{equation}
K(\psi)\geq\frac{1}{2r_{sch}^{(\psi)}-3}\mathrm{C}_{\rho_{\psi}}.\label{47}
\end{equation} 
Eqs. (\ref{38}), (\ref{47}) and relation $\Upsilon_{\rho_ {\psi}}
\geq\Upsilon_{chsh}^{(\rho_ {\psi})}$ prove the lower bound in (\ref{39_}). 
\end{proof}
\\

The lower bound  in (\ref{39_}) quantifies Bell nonlocality of a pure two-qudit state in terms of its entanglement and has no analogues in the literature.
This new bound explicitly proves 
the statement of Gisin and Peres \cite{09,010} that every entangled two-qudit state is Bell nonlocal.

The new lower bound in Theorem 4 and the upper bound (\ref{9__}), derived in \cite{019}, imply.

\begin{theorem}
    
For any pure two-qudit state $\rho_{\psi}=|\psi\rangle\langle\psi|$ 
 on   
$\mathbb{C}^{d_{1}}\otimes\mathbb{C}^{d_{2}}$,  the  maximal violation of all Bell inequalities, either on correlation functions or on joint probabilities and with any number of settings and outcomes at each of two
sites, admits the bounds
\begin{equation}
\sqrt{1+\frac{1}{(2r_{sch}^{(\psi)}-3)^{2}}%
\mathrm{C}_{\rho_{\psi}}^{2}}\leq\Upsilon_{\rho_{\psi}}\leq1+\sqrt{2r_{sch}^{(\psi)}(r_{sch}^{(\psi)}%
-1)}\mathrm{C}_{\rho_{\psi}},\label{39__}\
\end{equation}
where $1\leq r_{sch}^{(\psi)}\leq d=\min\{d_{1},d_{2}\}$ is the Schmidt rank of $|\psi\rangle$,  $\mathrm{C}_{\rho_{\psi}}$ is concurrence (\ref{5_}) of $\rho_{\psi}$ and
$\mathrm{C}_{\rho_{\psi}}=0\Leftrightarrow r_{sch}^{(\psi)}=1.$

\end{theorem}

From the upper bound and the lower bound it follows, respectively: $\Upsilon_{\rho}>1 \Rightarrow $  $\mathrm{C}_{\rho_{\psi}} > 0$  and $\mathrm{C}_{\rho_{\psi}} > 0 \Rightarrow \Upsilon_{\rho}>1$.

Therefore, since a pure two-qudit state $\rho$ is nonlocal iff  $\Upsilon_{\rho}>1$ 
 and entangled iff  $\mathrm{C}_{\rho_{\psi}} > 0$ (see in  Section 2), the upper bound and the lower bounds in (\ref{39__}) explicitly prove that entanglement of a pure two-qubit state is necessary and sufficient for its Bell nonlocality. For $d=2$, bounds in (\ref{39__}) reduce to those in (\ref{31}).

\section{Conclusion}

In the present article, for the maximal violation of all Bell inequalities by an arbitrary pure two-qudit state, we have found a new lower bound (Theorem 4), expressed via the concurrence of this state.  This new lower bound and the upper bound (\ref{9__}),  derived in \cite{019},  analytically quantify  Bell nonlocality of a pure two-qudit state via its entanglement (Theorem 5), in particular, prove explicitly that entanglement of a pure two-qudit state is necessary and sufficient for its Bell nonlocality. 

By
re-visiting a pure two-qubit case, we have also found and rigorously proved  (Theorems 1--3, Corollary 1) the new results on the  correlation properties and Bell nonlocality of an arbitrary  pure two-qubit state.

\section*{Acknowledgment}

The study was implemented in the framework of the Basic Research Program at
the National Research University Higher School of Economics (HSE University)
in 2023.


\begin{thebibliography}{99}        

\bibitem {01} Bell J S 1964 \emph{Phys. Phys. Fiz}. \textbf{1} 195

\bibitem {02}  Bell J S 1966 \emph{Rev. Mod. Phys.} \textbf{38} 447 

\bibitem {04}Loubenets E R 2008 \emph{J. Phys. A: Math. Theor.} \textbf{41} 445304

\bibitem {05}Loubenets E R 2012 \emph{J. Math. Phys.} \textbf{53} 022201

\bibitem {06}Brunner N, Cavalcanti D, Pironio S, Scarani V and Wehner S 2014
\emph{Rev. Mod. Phys}. \textbf{86} 419

\bibitem {07}Loubenets E R 2017 \emph{Found. Phys}. \textbf{47} 1100--1114

\bibitem {08} Werner R F 1989 \emph{Phys. Rev. A} \textbf{40} 4277

\bibitem {09}Gisin N 1991 \emph{Phys. Lett. A} \textbf{154} 201--202 

\bibitem {010}Gisin N and Peres A 1992 \emph{Phys. Lett. A }\textbf{162} 15

\bibitem {011}Clauser J F, Horne M A, Shimony A and Holt R A  1969 \emph{Phys.
Rev. Lett.} \textbf{23} 880--884

\bibitem {012}Horodecki R, Horodecki P, Horodecki M 1995 \emph{Phys. Lett. A
}\textbf{200} 340--344

\bibitem {019}Loubenets E R and Namkung M 2022 \emph{J. Phys. A: Math. Theor. }\textbf{55} 285301




\bibitem {014}Hill S A and Wootters W K  1997 \emph{Phys. Rev. Lett.
}\textbf{78} 5022--5025

\bibitem {014_}Rungta P, Buzek V, Caves C M, Hillery M and Milburn G J  2001 \emph{Phys. Rev. A} \textbf{64} 042315

\bibitem {015}Chen K, Albeverio S and Fei S M 2005 \emph{Phys. Rev Lett}.
\textbf{95} 040504

\bibitem {016}Kim J S, Das A and Sanders B C  2009 \emph{Phys. Rev. A}
\textbf{79} 012329




\bibitem {nlm5} Luo S  and Fu S 2011 \emph{Phys. Rev. Lett.}
\textbf{12} 120401

\bibitem {nlm2}De Vicente J I  2014 \emph{J. Phys. A: Math. Theor}
\textbf{47} 424017

\bibitem {nlm3} Fonseca E A and Parisio F 2015, \emph{Phys. Rev. A}
\textbf{92}, 030101.

\bibitem {nlm4} Hu M L and Fan H 2015, \emph{New J. Phys.}
\textbf{17}, 033004.

\bibitem {nlm1} Barasi{\'n}ski A and Nowotarski M 2018 \emph{Phys. Rev. A}
\textbf{98} 022132



\bibitem {018}Loubenets E R 2017 \emph{J. Math. Phys.} \textbf{58} 052202

\bibitem {CH}Clauser J F and Horne M A 1974 \emph{Phys. Rev. D} \textbf{10}
526

\bibitem {Moroder2013} Moroder T, Bancal J D, Liang Y C, Hofmann M and Guhne O 2013 \emph{Phys. Rev. Lett.} \textbf{111} 030501

\bibitem {Liang2011} Liang Y C, V{\'e}rtesi T and Brunner N 2011 
\emph{Phys. Rev. A}. \textbf{83}
022108

\bibitem {app}Goh K T, Bancal J D and Sacarani V 2016 \emph{New J. Phys.} \textbf{18}
045022

\bibitem {020}Verstraete F and Wolf M M 2002 \emph{Phys. Rev. Lett.}
\textbf{89} 170401

\bibitem {Batle}Batle J and Casas M 2011 \emph{J. Phys. A: Math. Theor}. \textbf{44}
445304

\bibitem {Gamel} Gamel O 2016 \emph{Phys. Rev. A}. \textbf{93}
062320 

\bibitem {21}Loubenets E R and Kulakov M 2021 \emph{J. Phys. A: Math. Theor.} \textbf{54}
195301




\end{thebibliography}
\end{document}